\documentclass[a4paper]{article}

\usepackage{times}
\usepackage{soul}
\usepackage{url}
\usepackage[hidelinks]{hyperref}
\usepackage[utf8]{inputenc}
\usepackage[small]{caption}
\usepackage{graphicx}
\usepackage{amsmath}
\usepackage{amsthm}
\usepackage{booktabs}
\usepackage{algorithm}
\newtheorem{theorem}{Theorem}
\newtheorem{lem}[theorem]{Lemma}
\newtheorem{prop}[theorem]{Proposition}
\newtheorem{corollary}[theorem]{Corollary}

\theoremstyle{definition}
\newtheorem{defn}[theorem]{Definition}
\newtheorem{remark}[theorem]{Remark}

\usepackage{multicol}
\newcommand{\circop}{\ominus}
\newcommand{\compa}{\lesseqgtr}
\newcommand{\landif}{\lanfull}
\newcommand{\ptype}{\type\infty}

\newcommand{\CIcon }{\CMod}

\newcommand{\CMod}{\mathfrak C}

\def\eqdef{\stackrel{\rm def}{=}}

\newcommand\cqm[1]{\nicefrac{\CMod}{#1}}

\newcommand{\landi}{\lanfull}

\usepackage{yfonts,array,amsfonts,amssymb}
\usepackage{amsmath,amsthm,units,stmaryrd}
\usepackage{cleveref}
\usepackage{upgreek,xcolor,bm,tikz,stackrel}
\usepackage[inline, shortlabels]{enumitem}
\usepackage{graphicx}
\usepackage[all]{xy}
\usepackage{tikz-cd}
\usetikzlibrary{decorations.pathmorphing}
\usepackage{wasysym}
\usetikzlibrary{arrows}

\usepackage[numbers]{natbib}
\usepackage{algpseudocode}
\newcommand{\skipprf}[1]{}

\newcommand{\cl }{\mathcal }

\newcommand{\define}[1]{\textbf{#1}}
\newcommand{\diam}{\Diamond}
\newcommand{\ubox}{\Box}
\newcommand{\logbasic}{\gtl}
\newcommand{\tnext}{\ocircle}

\newcommand{\iiff}{\mathop\Leftrightarrow}
\newcommand{\ineg}{{\sim}}
\newcommand{\imp}{\mathop \Rightarrow}
\newcommand{\dimp}{\mathop \Leftarrow}

\newcommand{\gtl}{{\sf GTL}}

\newcommand{\lanfull}{{\mathcal L}}

\newcommand{\ignore}[1]{}

\newcommand{\lgt}[1]{\|#1\|}
\newcommand{\nx}{{\ocircle}}
\newcommand{\nec}{\Box}
\newcommand{\ps}{\Diamond}
\newcommand{\val}[1]{\lb #1 \rb}

\newcommand{\type}[1]{\mathbb T_{ #1 }}
\newcommand{\peq}{\leq}

\def\lb{\left\llbracket}
\def\rb{\right\rrbracket}
\def\<{\left (}

\def\>{\right )}
\def\({\left (}
\def\){\right )}

\def\cbra{\left \{}
\def\cket{\right \}}

\DeclareSymbolFont{AMSb}{U}{msb}{m}{n}
\DeclareMathSymbol{\N}{\mathbin}{AMSb}{"4E}
\DeclareMathSymbol{\Z}{\mathbin}{AMSb}{"5A}
\DeclareMathSymbol{\R}{\mathbin}{AMSb}{"52}
\DeclareMathSymbol{\Q}{\mathbin}{AMSb}{"51}
\DeclareMathSymbol{\I}{\mathbin}{AMSb}{"49}
\DeclareMathSymbol{\C}{\mathbin}{AMSb}{"43}

\usepackage{authblk}

\title{A G\"odel Calculus for Linear Temporal Logic}
\author[1,2]{Juan Pablo Aguilera \footnote{\href{mailto:juan.aguilera@UGent.be}{\tt juan.aguilera@UGent.be}}}
\author[3]{Mart\'in Di\'eguez \footnote{\href{mailto:martin.dieguezlodeiro@univ-angers.fr}{\tt martin.dieguezlodeiro@univ-angers.fr}}}
\author[1,4]{David Fern\'andez-Duque \footnote{\href{mailto:david.fernandezduque@ugent.be}{\tt david.fernandezduque@ugent.be}}}
\author[1]{Brett McLean \footnote{\href{mailto:brett.mclean@ugent.be}{\tt brett.mclean@ugent.be}}}

\affil[1]{\centering Department of Mathematics WE16, Ghent University, Ghent, Belgium}
\affil[2]{\centering Institute of Discrete Mathematics and Geometry, Vienna University of Technology, Vienna, Austria}
\affil[3]{\centering LERIA, University of Angers, Angers, France}
\affil[4]{\centering ICS of the Czech Academy of Sciences, Prague, Czech Republic}

\begin{document}

\maketitle

\begin{abstract}
We consider G\"odel temporal logic ($\sf GTL$), a variant of linear temporal logic based on G\"odel--Dummett propositional logic.
In recent work, we have shown this logic to enjoy natural semantics both as a fuzzy logic and as a superintuitionistic logic.
Using semantical methods, the logic was shown to be {\sc pspace}-complete.
In this paper we provide a deductive calculus for $\sf GTL$, and show this calculus to be sound and complete for the above-mentioned semantics.
\end{abstract}

\section{Introduction}


Despite their potential usefulness in areas such as spatio-temporal reasoning \cite{arte} or vague temporal reasoning \cite{Davies96}, the combination of linear temporal logic with a modal or non-classical base tends to lead to high computational complexity \cite{Balbiani2017} or even undecidability \cite{konev,Vidal21}.
Nevertheless, our recent work
~\cite{gtllics} provides a promising avenue for fuzzy temporal reasoning.
There, we show that linear temporal logic ($\sf LTL$) over a G\"odel--Dummett base is {\sc pspace}-complete, which is optimal for a logic interpreting classical $\sf LTL$.

The methods used in \cite{gtllics} are model-theoretic and leave open the question of whether a proof-theoretic approach is possible. Here, we aim to close this gap by employing techniques used to establish completeness of an intuitionistic $\sf LTL$ in \cite{Boudou2017} with `eventually' but without `henceforth.'
As we will see, the complications that led to omitting `henceforth' in that work can be solved by incorporating the {\em dual implication} into our language, a connective that is to implication what disjunction is to conjunction.
Aside from the technical advantages it afford us, it has been argued by \citeauthor{Ra80}~\cite{Ra80} that dual implication is useful for reasoning with incomplete or inconsistent information.

Previously, $\sf LTL$ based on the intermediate logic of here-and-there~\cite{heyting30a} was axiomatized by \citeauthor{BalbianiDieguezJelia}~\cite{BalbianiDieguezJelia}.
This logic allows for three truth values and is the basis for temporal answer set programming~\cite{taspa,taspb}. Another combination of here-and-there, modal logic and Rauszer's co-implication has been studied in~\cite{BD18}. 

G\"odel logics and their extensions with possibility theory~\cite{DuboisLP87} have been extensively studied in~\cite{DELLUNDE201163}.
These extensions have applications in the field of logic programming~\cite{AG13,BlandiGR05}.
Aside from this, a version of the `next' fragment of intuitionistic $\sf LTL$ was axiomatized by \cite{KojimaNext} and a logic with `next' and `eventually' (but not henceforth) by~\cite{Boudou2017}.
Intuitionistic $\sf LTL$ with `henceforth' has not been axiomatized, but \cite{jungle} showed that logics with the latter tense are more sensitive to choice of semantics than those without it and \cite{ChopoghlooM21} provided a strongly complete infinitary calculus.

We recently showed that $\sf GTL$ possesses two natural semantics, corresponding to whether it is viewed as a fuzzy logic or as a superintuitionistic logic~\cite{gtllics}.
As a fuzzy logic, propositions take values in $[0,1]$, and truth values of compound propositions are defined using standard operations on the reals.
As a superintuitionistic logic, models consist of bi-relational structures equipped with a partial order to interpret implication intuitionistically and a function to interpret the $\sf LTL$ tenses.
We showed that the set of validities for either of these semantics coincides, and in fact coincides with the set of validities for a third class of structures we call {\em non-deterministic quasimodels.}
Similar structures were used to prove upper complexity bounds for dynamic topological logic \cite{Fernandez09} and intuitionistic temporal logic \cite{F-D18}.
In the setting of $\sf GTL$, they can be used to prove that the validity problem is decidable: while the logic does not enjoy the finite model property for either the fuzzy or the superintuitionistic semantics, it {\em does} enjoy the finite quasimodel property.

\cite{eventually,dtlaxiom} have shown that quasimodels also come in handy in completeness proofs.
There are two main reasons for this.
First, as quasimodels are somewhat more flexible than proper models, it is easier to construct them.
Thus our task is to construct a quasimodel falsifying a given non-derivable formula.
Once constructed, we can use unwinding techniques to produce a proper bi-relational model from our quasimodel.
The second advantage is that it allows us to use techniques normally available only for logics enjoying the finite model property, such as fully characterizing a given structure using a finite formula.
In fact, contrary to the classical case, we will assign two characteristic formulas to each state $w$ of our quasimodels, $\chi^+(w)$ and $\chi^-(w)$, characterizing the `positive' and `negative' information available in $w$.

Aside from these points, our completeness proof involves several stages that should be standard to those familiar with temporal logic: a {\em canonical model} $\CMod$ is built, which correctly interprets $\nx$ but not $\diam$ and $\Box$.
In order to remedy this, a modified filtration $\cqm\Sigma$ is built, which {\em does} interpret all tenses correctly when restricted to the set of formulas $\Sigma$, but fails to have a deterministic successor relation.
Finally, an unwinding procedure is applied to $\cqm\Sigma$ to obtain a bi-relational model.
This unwinding is somewhat more complex than in the classical case, but fortunately we may appeal to results of \cite{gtllics} to conclude that any formula falsified in $\cqm\Sigma$ is also falsified in a suitable unwinding of $\cqm\Sigma$.

\ignore{
\subsubsection*{Structure of the paper} In \Cref{SecBasic} we introduce the temporal language that we work with, and then introduce both the real semantics and bi-relational semantics for G\"odel temporal logic.
In \Cref{secAx} we list our set of axioms and establish soundness.
In \Cref{SecNDQ} we recall the definitions of quasimodels and some useful results of \cite{gtllics}.
\Cref{secCan} defines the canonical model and \Cref{Sec:quotient} the canonical quasimodel, obtained by filtration.
In order to prove that the latter is indeed a quasimodel, \Cref{SecChar} introduces characteristic formulas and \Cref{SecComp} uses them to establish that the conditions on `henceforth' and `eventually' hold on the canonical quasimodel, finishing our completeness proof.
\Cref{SecConc} provides some concluding remarks.
}

\section{Syntax and Semantics}\label{SecBasic}

In this section we first introduce the temporal language we work with and then two possible semantics for this language: real semantics and bi-relational semantics. 

Fix a countably infinite set $\mathbb P$ of propositional variables. Then the \define{G\"odel temporal language} $\lanfull$ is defined by the grammar (in Backus--Naur form):
\[\varphi,\psi :=   \   p  \ |  \ \varphi\wedge\psi \  |  \ \varphi\vee\psi  \ |  \ \varphi\imp \psi  \ |  \  \ \varphi\dimp \psi  \ |  \ \nx\varphi \  | \  \ps\varphi \  |  \ \nec\varphi , \]
where $p\in \mathbb P$. Here, $\nx$ is read as `next', $\ps$ as `eventually', and $\nec$ as `henceforth'. The connective $\dimp$ is dual (or co-) implication and represents the operator dual to implication \cite{Wolter1998}. 
 We also use $\bot$ 
  (respectively~$\top$) as a shorthand for $p\dimp p$ (respectively~$p\imp p$) for some fixed variable $p$, and we use $\neg\varphi$ (respectively~$\ineg \varphi$) as a shorthand for $\varphi\imp \bot$ (respectively~$\top \dimp \varphi$), and $\varphi\iiff \psi$ (not related to dual implication) as a shorthand for $(\varphi\imp \psi) \wedge (\psi\imp\varphi)$.


	We now introduce the first of our semantics for the G\"odel temporal language: \emph{real semantics}, which views $\lanfull$ as a fuzzy logic (enriched with temporal modalities). In the definition, $[0,1]$ denotes the real unit interval.

\begin{defn}[real semantics]\label{DefRSem}
A \define{flow} is a pair $\mathcal T = (T,S)$,
where $ T $ is a set and $S \colon T \to T$ is a function.
A \define{real valuation} on $\mathcal T$ is a function $V \colon \lanfull \times T \to \mathbb [0,1]$ such that, for all $t\in T$, the following equalities hold.
\[
\begin{array}{rcl}
V(\bot,t)&=& 0 \\
V(\varphi\wedge\psi,t) &=& \min \{ V(\varphi ,t) , V( \psi,t) \}\\
V(\varphi \vee \psi,t) &=& \max \{ V(\varphi ,t) , V( \psi,t) \}\\
V(\varphi \imp \psi,t) &=&\begin{cases}
V(\psi,t )&\text{if }V(\varphi,t ) > V(\psi,t )\\
1&\text{if }V(\varphi, t)\leq V(\psi, t)

\end{cases}
\\
V(\varphi \dimp \psi,t) &=&\begin{cases}
V(\varphi,t )&\text{if }V(\varphi,t ) > V(\psi,t )\\
0&\text{if }V(\varphi, t)\leq V(\psi, t)

\end{cases}
\\
V(\nx\varphi, t)&=&V( \varphi,S(t))\\
V(\ps\varphi,t)&=& \sup_{n<\omega}  V(\varphi,S^n(t))\\
V(\nec\varphi,t)&=& \inf_{n<\omega}  V(\varphi,S^n(t))\\
\end{array}
\]
A flow $\mathcal  T$ equipped with a valuation $V$ is a \define{real (G\"odel temporal) model}.
\end{defn}

The second semantics, \emph{bi-relational semantics}, views $\lanfull$ as an intuitionistic logic (temporally enriched).

\begin{defn}[bi-relational semantics]\label{DefKSem}
A \define{(G\"odel temporal) bi-relational frame} is a quadruple $\mathcal  F=(W,T,{\leq},S)$ where $(W,\leq)$ is a linearly ordered set and $(T,S)$ is a flow.
A \define{bi-relational valuation} on $\mathcal  F$ is a function $\lb\cdot\rb\colon\lanfull \to 2^{W\times T}$ such that, for each $p \in \mathbb P$, the set $\lb p \rb$ is \emph{downward closed} in its first coordinate, and the following equalities hold.
\[
\begin{array}{rcl}
\lb\bot\rb&=&\varnothing \\
\lb\varphi\wedge\psi\rb &=&\lb\varphi\rb\cap \lb\psi\rb\\
\lb\varphi\vee\psi\rb &=&\lb\varphi\rb\cup \lb\psi\rb\\
\lb\varphi\imp\psi\rb &=& \{ (w,t) \in W\times T \mid \forall v\leq w ((v,t)\in \lb\varphi \rb  \\
&&\hspace{2.75cm}\text{ implies } (v,t)\in \lb\psi \rb ) \}\\
\lb\varphi\dimp\psi\rb &=& \{ (w,t) \in W\times T \mid \exists v\geq w ((v,t)\in \lb\varphi \rb  \\
&&\hspace{2.75cm}\text{ and } (v,t)\notin \lb\psi \rb ) \}\\
\val{\nx\varphi}&=&(\mathrm{id}_W \times S)^{-1} \val\varphi\\
\val{\ps\varphi}&=& \bigcup_{n<\omega}(\mathrm{id}_W \times S)^{-n} \val\varphi\\
\val{\nec\varphi}&=& \bigcap_{n<\omega}(\mathrm{id}_W \times S)^{-n}\val\varphi  \\
\end{array}
\]
\noindent where $(\mathrm{id}_W \times S)$ is the function such that $(\mathrm{id}_W \times S)(w,t) = (w,S(t))$.
Given $(w,t) \in W \times T$, we say that $S((w,t)) \eqdef (w,S(t))$. 
A bi-relational frame $\mathcal  F$ equipped with a valuation $\lb\cdot\rb$ is a \define{(G\"odel temporal) bi-relational model}.
\end{defn}

This semantics combines standard semantics for the implications based on $\leq$ (read downward) and for the tenses based on $S$: for example, $(w,t)\in \val{\ps\varphi}$ if and only if there exists $n\geq 0$ such that $(w,S^n(t)) \in \val\varphi$.
Note that, by structural induction, the valuation of \emph{any} $\varphi \in \lanfull$ is downward closed in its first coordinate, in the sense that if $(w,t)\in \val\varphi$ and $v\leq w$, then $(v,t)\in \val \varphi$.


Validity of $\lanfull$-formulas is defined in the usual way.

\begin{defn}[validity]
Given a real model $\mathcal  X = (T, S, V)$ and a formula $\varphi\in \lanfull$, we say that $\varphi$ is \define{globally true} on $\mathcal  X$, written $\mathcal  X\models\varphi$, if for all $t \in T$ we have $V(\varphi, t) = 1$.
Given a bi-relational model $\mathcal  X = (\mathcal F, \lb\cdot\rb)$ and a formula $\varphi\in \lanfull$, we say that $\varphi$ is \define{globally true} on $\mathcal  X$, written $\mathcal  X\models\varphi$, if $\val\varphi =W \times T$. 

If $\mathcal  X$ is a flow or a bi-relational frame, we write $\mathcal  X\models\varphi$  and say $\varphi$ is \define{valid} on $\mathcal X$, if $\varphi$ is globally true for every valuation on $\mathcal  X$. If $\Omega$ is a class of flows, frames, or models, we say that $\varphi\in\lanfull$ is \define{valid} on $\Omega$ if, for every $\mathcal  X\in \Omega$, we have $\mathcal  X\models\varphi$. If $\varphi$ is not valid on $\Omega$, it is \define{falsifiable} on $\Omega$.
\end{defn} 

We define the logic $\gtl_{\mathbb R}$ to be the set of  $\lanfull$-formulas that are valid over the class of all flows and the logic $\gtl_\mathrm{Rel}$ to be the set of $\lanfull$-formulas that are valid over the class of all bi-relational frames.
The main theorem of \cite{gtllics} is that these logics coincide:

\begin{theorem}\label{theoLICSone}
$\gtl_{\mathbb R} = \gtl_\mathrm{Rel}$. 
That is, for each $\varphi\in\lanfull$, $\varphi$ is valid over the class of real G\"odel temporal models if and only if it is valid over the class of all G\"odel temporal bi-relational models.
\end{theorem}

\skipprf{
To give some intuition for Theorem \ref{theoLICSone}, first assume that $\varphi$ is falsifiable on a real model $(T,S,V)$.
We could attempt to view this as a bi-relational model on the frame $(\mathbb R,T,\leq,S)$, where $\leq$ is the usual order on the real numbers.
In fact this does not quite work but can be adjusted by instead considering $W = \mathbb R\setminus A$ for a suitable countable set $A$.
Conversely, if $\varphi$ is falsifiable on a bi-relational model, we can use a L\"owenheim--Skolem argument to see that $\varphi$ is falsifiable on a \emph{countable} bi-relational model $(W,T,\leq,S,\val\cdot)$, then use the fact that any countable order is embeddable into the real line, in order to embed $W$ into $\mathbb R$.
With the aid of this embedding, we can then construct a real model falsifying $\varphi$.
For details (of both directions), we refer the reader to \cite{gtllics}.
}

\section{The calculus}\label{secAx}

We begin by establishing our basic calculus for logics over $\mathcal L$.
It is obtained by adapting the standard axioms and inference rules of $\sf LTL$ \cite{temporal}, as well as their dual versions.

\begin{defn}\label{defLogbasic}
The logic $\logbasic$ is the least set of $\mathcal L$-formulas closed under the following axioms and rules.

\begin{enumerate}[wide, labelwidth=!, labelindent=0pt,label={\sc \Roman*},ref={\sc\roman*}]
\item\label{ax00Intu} {\bf All (substitution instances of) intuitionistic tautologies (see e.g.~\cite{MintsInt})}
\item\label{ax01Taut}{\bf Axioms and rules of H-B logic:}
\begin{enumerate}[label={\sc\alph*},ref={{\sc\roman{enumi}}.{\sc\alph*}}]
\item \label{axco01} $\varphi \imp \left(\psi \vee \left(\varphi \dimp \psi\right)\right)$


\item\label{axDimpMon} $\dfrac{\varphi \imp \psi }{(\varphi \dimp \theta) \imp (\psi \dimp \theta)}$

\item\label{axDimpDis} $\dfrac{\varphi \imp \psi \vee \gamma}{(\varphi \dimp \psi) \imp \gamma}$

\end{enumerate}

\item \textbf{Linearity axioms:}
\begin{enumerate}[label={\sc \alph*},ref={{\sc\roman{enumi}}.{\sc\alph*}}]
\item \label{axGodel} $\left(\varphi \imp \psi\right) \vee \left(\psi \imp \varphi\right)$ 
\item \label{axcoGodel} $\neg \left(\left(\varphi \dimp \psi\right) \wedge \left(\psi \dimp \varphi\right)\right)$ 
\end{enumerate}

\item \textbf{Temporal axioms: }
\begin{enumerate}[label={\sc \alph*},ref={{\sc\roman{enumi}}.{\sc\alph*}}]
\item\label{ax02Bot} $\neg \tnext \bot$
\item\label{ax04NexVee} $\tnext \left( \varphi \vee \psi \right) \imp \left(\tnext \varphi \vee\tnext \psi\right)$
\item\label{ax03NexWedge} $ \left(\tnext \varphi \wedge\tnext \psi\right) \imp \tnext \left( \varphi \wedge \psi \right) $
\item\label{ax05KNext} $\tnext\left( \varphi \imp \psi \right) \iiff \left(\tnext\varphi \imp \tnext\psi\right)$
\item\label{ax06KBox} $\ubox \left( \varphi \imp \psi \right) \imp \left(\ubox \varphi \imp \ubox \psi\right)$
\item\label{ax07:K:Dual} $\ubox \left( \varphi \imp \psi \right) \imp \left(\diam \varphi \imp \diam \psi\right)$
\item\label{ax09BoxFix} $\ubox \varphi \imp  \varphi\wedge \tnext \ubox \varphi$

\item\label{ax10DiamFix} $\varphi\vee \tnext \diam \varphi \imp \diam \varphi$
\item\label{ax11:ind:1} $\ubox ({ \varphi \imp \tnext \varphi } )\imp ({ \varphi \imp \ubox \varphi })$
\item\label{ax12:ind:2} $\ubox ({ \tnext \varphi \imp \varphi})\imp ({ \diam \varphi \imp \varphi } )$
\end{enumerate}

\item\label{ax:backward:dual} \textbf{Back--up confluence axiom: }

$ \tnext \left(\varphi \dimp \psi\right) \imp \left(\tnext \varphi \dimp \tnext  \psi\right) $

\item \textbf{Standard modal rules:}
\begin{multicols}3
\begin{enumerate}[label={\sc \alph*},ref={{\sc\roman{enumi}}.{\sc\alph*}}]
\item\label{ax13MP}  
$\dfrac{\varphi,\ \varphi\imp \psi}\psi$
\item\label{ax14NecCirc} $\dfrac\varphi {\tnext\varphi}$
\item\label{ax14NecBox}  $\dfrac\varphi {\ubox\varphi}$
\end{enumerate}
\end{multicols}
\end{enumerate}
\end{defn}

Axiom group \ref{ax01Taut} concerns the relationship between $\imp$ and $\dimp$. In particular, Axiom \ref{axco01} is used in C. Rauszer's axiomatization of intuitionistic logic with co-implication (called H-B logic)~\cite{Rauszer74}. The G\"odel--Dummett axiom~\ref{axGodel} and its order dual \ref{axcoGodel} are used to force the connectives $\imp$ and $\dimp$ to be implemented on locally linear posets (i.e. posets that are a disjoint union of linear orders). 

Axioms \ref{ax03NexWedge}, \ref{ax05KNext}, and \ref{ax06KBox} are standard modal axioms (viewing $\nx$ as a box-type modality). In particular they hold in any normal modal logic, although of course $\gtl$ is not itself normal by virtue of being strictly sub-classical. Axiom \ref{ax07:K:Dual} is a dual version of \ref{ax06KBox}; such dual axioms are often needed in intuitionistic modal logic, since $\diam$ and $\ubox$ are not typically interdefinable.
The axioms \ref{ax02Bot} and \ref{ax04NexVee} have to do with the passage of time being deterministic in linear temporal logic:  \ref{ax02Bot} characterises seriality and  \ref{ax04NexVee} characterises (partial) functionality, thus together they constrain temporal accessibility to be a total function. 

The co-inductive axiom \ref{ax09BoxFix} states that if something will henceforth be the case, then it is true now and, moreover, in the next moment, it will still henceforth be the case, and \ref{ax11:ind:1} is successor induction, as time is interpreted over the natural numbers. Axioms \ref{ax10DiamFix} and \ref{ax12:ind:2} are their duals. Note that `henceforth' is interpreted reflexively. 
All rules of group \ref{ax13MP} are standard modal logic deduction rules, and in particular any normal modal logic is closed under these rules.

Most of the axioms are either included in the axiomatization of intuitionistic $\sf LTL$ \cite{Boudou2017} or a variant of one of them (e.g. a contrapositive). 
From this, we easily derive the following.

\begin{prop}
The above calculus is sound for the class of real models, as well as for the class of bi-relational models.
\end{prop}

\begin{proof}
The rules \ref{axDimpMon} and \ref{axDimpDis} are readily seen to preserve validity.
We check Axioms \ref{ax:backward:dual} and~\ref{axcoGodel}; all other rules or axioms have been shown to be sound for intuitionistic or bi-relational models in the literature (see e.g.~\cite{Balbiani2017,Rauszer74}).

For Axiom \ref{ax:backward:dual}, it suffices to check its validity on the class of bi-relational models.
Let $\mathcal M=(W,T,\leq,S,\val\cdot)$ be a bi-relational model and suppose that $(w,t)\in \val{\tnext \left(\varphi \dimp \psi\right) }$; we must show that $(w,t)\in \val{\tnext  \varphi \dimp \tnext \psi  }$.
From $(w,t)\in \val{\tnext \left(\varphi \dimp \psi\right) }$ we see that $(w,S(t))\in \val{ \varphi \dimp \psi  }$; hence there is $(v,S(t))\geq (w,S(t))$ with $(v,S(t))\in \val\varphi\setminus\val\psi$.
But then $(v,S(t)) \in \val{\tnext\varphi}\setminus\val{\tnext\psi}$, witnessing that $(w,t)\in \val{\tnext\varphi\dimp\tnext\psi}$.

\noindent For Axiom \ref{axcoGodel}, let us assume towards a contradiction that $\neg \left(\left( \varphi \dimp \psi\right) \wedge \left(\psi \dimp \varphi\right)\right)$ is not valid with respect to bi-relational models, so we can find $\mathcal{M}$ as above and $(w,t)\in\mathcal{M}$ such that $(w,t)\not \in \val{\neg \left(\left( \varphi \dimp \psi\right) \wedge \left(\psi \dimp \varphi\right)\right)}$.
Therefore, there exists $(v,t) \le (w,t)$ such that $(v,t) \in \val{\left( \varphi \dimp \psi\right) \wedge \left(\psi \dimp \varphi\right)}$.
Therefore, $(v,t) \in \val{\varphi \dimp \psi}$ and  $(v,t) \in \val{\varphi \dimp \psi}$.
Therefore, there exists $(v',t)\ge (v,t)$ and $(v'',t)\ge (v,t)$ such that $(v',t) \in \val{\varphi}\setminus \val{\psi}$ and $(v'',t) \in \val{\psi}\setminus\val{\varphi}$. 
Since $(W,\le)$ is a linear order, either $(v',t) \le (v'',t)$ or $(v',t) \ge (v'',t)$.
In the former case we get that $(v',t) \in \val{\psi}$ and in the latter case we get that $(v'',t) \in \val{\varphi}$; in any case we reach a contradiction.
\end{proof}

Our main objective is to show that our calculus is indeed complete; proving this will take up the remainder of this paper.

As we show next, we can also derive the converses of some of these axioms. Below, for a set of formulas $\Gamma$ we define $\tnext \Gamma = \{\tnext\varphi : \varphi \in \Gamma\}$, and empty conjunctions and disjunctions are defined by $\bigwedge\varnothing =\top$ and $\bigvee \varnothing = \bot$.

\skipprf{When reasoning about $\gtl$, it is useful to note that by \ref{ax00Intu} and \ref{ax13MP}, the following weak form of the deduction theorem holds: for $\phi, \psi \in \lanfull$, we have $\phi \imp \psi \in \gtl$ if (and only if) we can deduce $\psi$ from $\phi$ in the system with $\gtl$ as its axioms and modus ponens as its only deduction rule. We will often implicitly use this to simplify reasoning.}

\begin{lem}\label{lemmReverse}
	Let $\varphi \in \landi$ and $\Gamma\subseteq \landi$ be finite. Then the following formulas belong to $\gtl$.
	\begin{multicols}2
	\begin{enumerate}
		
		\item $\tnext \bigvee \Gamma \Leftrightarrow \bigvee \tnext \Gamma$
		
		\item $\tnext \bigwedge \Gamma \Leftrightarrow \bigwedge \tnext \Gamma$

		\item\label{itReverseDiam} $\diam \varphi \imp \varphi \vee \tnext \diam \varphi$
		\item\label{itReverseBox} $\varphi \wedge \tnext \nec \varphi \imp \nec \varphi$
		
		\item \label{itNotRef} $ (\varphi\dimp\varphi) \imp \psi$

		\item\label{itReverseDimp} $(\varphi \dimp \psi) \imp \varphi$
	\end{enumerate}
\end{multicols}	
\end{lem}

\color{black}

\section{Labelled systems and quasimodels}\label{SecNDQ}

Quasimodels will be a central tool in our completeness proof.
These were originally introduced in \cite{Fernandez09} for \emph{dynamic topological logic}, a classical predecessor of \emph{intuitionistic temporal logic}, for which quasimodels were also used in \cite{F-D18}.
In this section we will introduce labelled spaces, labelled systems, and finally, quasimodels. Quasimodels can be viewed as a sort of nondeterministic generalisation of bi-relational models.
Quasimodels are a great advantage to us since $\gtl$ has the finite quasimodel property (any falsifiable formula is falsifiable in a finite quasimodel), despite not having the finite model property for either the real or the bi-relational semantics \cite{gtllics}.

\begin{defn}\label{def:type}
	Let $\Sigma \subseteq \landif$ be closed under subformulas and $\Phi^+,\Phi ^ - \subseteq \Sigma$. We say that the pair $\Phi = (\Phi ^+ , \Phi ^-) $ is a \define{two-sided $\Sigma$-type} if:
	\begin{enumerate}
		\item $\Phi^- \cap \Phi ^ +  = \varnothing$,
		\label{cond:type:intersection}
		
		
		\item if $\varphi\wedge\psi\in \Phi ^ +$, then $\varphi,\psi\in \Phi^+$,
		\label{cond:type:posconj}
		
		\item if $\varphi\wedge\psi\in \Phi ^ -$, then	
		$\varphi \in \Phi ^-$ or $\psi\in \Phi^-$,
		\label{cond:type:negconj}
		
		\item if $\varphi\vee\psi\in \Phi ^ +$, then	$\varphi \in \Phi^+$ or $\psi\in \Phi^+$,
		
		\label{cond:type:posdisj}
		
		\item if $\varphi\vee\psi\in \Phi ^ -$, then  $\varphi , \psi\in \Phi^-$,
		\label{cond:type:negdisj}
		
		\item if $\varphi\imp\psi\in \Phi^+$, then $\varphi \in \Phi^-$ or $\psi \in\Phi^+$,
		\label{cond:type:implication}
		
		\item if $\varphi\imp\psi\in \Phi^-$, then $\psi \in\Phi^-$, 
		\label{cond:type:implication:neg} 
		
		\item if $\varphi\dimp \psi\in\Phi^-$, then $\varphi  \in \Phi^-$ or $\psi \in \Phi^+$, \label{cond:type:coimplication:neg}
		
		\item if $\varphi \dimp \psi \in\Phi^+$, then $\varphi \in \Phi^+$, \label{cond:type:coimplication}

		\item\label{cond:type:diam} if $\diam \varphi \in \Phi^-$, then $\varphi \in \Phi^-$,
		\item\label{cond:type:box} if $\nec \varphi \in \Phi^+$, then $\varphi \in \Phi^+$.
		
	\end{enumerate}
	If moreover $\Sigma = \Phi^- \cup \Phi^+$, we may say that $\Phi$ is \define{saturated}.
	The set of saturated two-sided $\Sigma$-types will be denoted $\type{\Sigma}$.  Given $\Phi, \Psi \in \type{\Sigma}$, we write
	\[\Phi\leq_\Sigma \Psi \text{ if and only if  } \Phi^- \subseteq \Psi^- \hbox{ and } \Phi^+ \supseteq \Psi^+.\] 
\end{defn}

 Often we want $\Sigma$ to be finite, in which case we write $\Sigma\Subset \lanfull$ to indicate that $\Sigma\subseteq \lanfull$ and $\Sigma$
is finite and closed under subformulas.
We remark that if $\Phi\in \type\Sigma$, then $\Phi^-=\Sigma\setminus \Phi^+$ (and vice-versa), but it is convenient to view $\Phi$ as a pair, since both the `positive' and `negative' information will play an important role.

A partially ordered set $(A,\leq)$ is \define{locally linear} if it is a disjoint union of linear posets.
If $a,b\in A$, we write $a\compa b$ if $a\leq b$ or $b\leq a$.
We call the set $\{b\in A:b\compa a\}$ the \define{linear component} of $a$; by assumption, linear components partition $A$.

\begin{defn}\label{frame}
Let $\Sigma\subseteq\lanfull$ be closed under subformulas.
A \define{$\Sigma$-labelled space} is a triple $\mathcal  W= ( |\mathcal  W|,\leq_\mathcal  W ,\ell_\mathcal  W )$, where $( |\mathcal  W| ,\leq_\mathcal  W )$ is a locally linear poset and $\ell\colon |\mathcal  W| \to \type\Sigma$ a monotone function, in the sense that 
\[w\leq_\mathcal W v \text{ implies } \ell_\mathcal  W(w) \leq_\Sigma \ell_\mathcal  W(v),\]
and such that for all $w\in |\mathcal  W|$:
\begin{itemize}\item
 whenever $\varphi\imp \psi\in  \ell^-_\mathcal  W(w) $, there is $v\leq _\mathcal W w$ such that $\varphi\in \ell^+_\mathcal  W(v)$ and $\psi  \in \ell^-_\mathcal  W(v)$;
 \item
 whenever $\varphi\dimp \psi\in \ell^+_\mathcal  W(w) $, there is $v\geq_\mathcal W  w$ such that $\varphi\in \ell^+_\mathcal  W(v)$ and $\psi\in \ell^-_\mathcal  W (v)$.
 \end{itemize}

The $\Sigma$-labelled space $\mathcal  W$ \define{falsifies} $\varphi\in\lanfull$ if $\varphi\in \ell_\mathcal  W(w)^-$ for some $w\in W$.
The \define{height} of $\mathcal W$ is the supremum of all $n$ such that there is a chain $w_1 <_\mathcal  W w_2 <_\mathcal  W \ldots <_\mathcal  W  w_n$.
\end{defn}

If $ \mathcal W$ is a labelled space, elements of $|\mathcal W|$ will sometimes be called \define{worlds}.
When clear from context we will omit subscripts and write, for example,~$\leq$ instead of $\leq_\mathcal  W$. 

Recall that a subset $S$ of a poset $(P,\leq)$ is \define{convex} if $s \in S$ whenever $a,b \in S$ and $a\leq s\leq b$.
A \define{convex relation} 
 between posets $(A,\leq_A)$ and $(B,\leq_B)$ is a binary relation $R\subseteq A\times B$ such that for each $x \in A$ the image set $\{y \in B \mid x \mathrel R y\}$ is convex with respect to $\leq_B$, and for each $y \in B$ the preimage set $\{x \in A \mid x \mathrel R y\}$ is convex with respect to $\leq_A$.
 The relation $R$ is \define{fully confluent} if it validates the four following conditions:
\begin{description}

\item[Forth--down]\label{forward--down}\phantom{.}\newline if $x \leq _A x' \mathrel R y'$ there is $y$ such that $x \mathrel R y \leq_B y'$,

\item[Forth--up]\label{forward--up}\phantom{.}\newline if $x' \geq _A x \mathrel R y$ there is $y'$ such that $x' \mathrel R y' \geq_B y$, 

\item[Back--down]\label{backward--down}\phantom{.}\newline if $x' \mathrel R y' \geq_B y$ there is $x$ such that $x'  \geq _A x \mathrel R y$,

\item[Back--up]\phantom{.}\newline if $x \mathrel R y \leq_B y'$ there is $x'$ such that $x \leq _A x' \mathrel R y'$.

\end{description}
In other words, $R$ is fully confluent if ${\leq_A} \circ R = R \circ {\leq_B}$ and ${\geq_A} \circ R = R \circ {\geq_B}$.

\begin{defn}\label{compatible}
Let $\Sigma\subseteq\lanfull$ be closed under subformulas. Suppose that $\Phi,\Psi\in\type\Sigma$. The ordered pair $(\Phi,\Psi)$ is \define{sensible} if it satisfies the following conditions:
\begin{enumerate}
\item If $\nx\varphi\in \Phi^+$, then $ \varphi\in \Psi^+$.
\item If $\nx\varphi\in \Phi^-$, then $ \varphi\in \Psi^-$.
\item If $\ps\varphi\in \Phi^+$, then $\varphi\in\Phi^+$ or $\ps\varphi\in \Psi^+$. 
\item If $\ps\varphi\in \Phi^-$, then $\varphi\in\Phi^-$ and $\ps\varphi\in \Psi^-$.
\item If $\nec \varphi\in \Phi^+$, then $\varphi\in \Phi^+$ and $\nec \varphi\in \Psi^+$.
\item If $\nec \varphi\in \Phi^-$, then $\varphi\in \Phi^-$ or $\nec \varphi\in \Psi^-$.
\end{enumerate}
A pair $(w,v)$ of worlds in a labelled space $\mathcal  W$ is \define{sensible} if $(\ell (w),\ell (v))$ is sensible.
A relation
$S\subseteq |\mathcal  W|\times |\mathcal  W|$
is \define{sensible} if every pair in $S$ is sensible.
Further, $S$ is \define{$\omega$-sensible} if
\begin{itemize}

\item whenever $\ps\varphi\in \ell^+_\mathcal W(w)$, there are $n\geq 0$ and $v$ such that $w \mathrel S^n v$ and $\varphi\in \ell^+_\mathcal W(v)$; 

\item whenever $\nec\varphi\in \ell^-_\mathcal W(w)$, there are $n\geq 0$ and $v$ such that $w \mathrel S^n v$ and $\varphi\in \ell^-_\mathcal W(v)$.

\end{itemize}
Recall that a binary relation is said to be \define{serial} if every element of the domain is related to some element of the co-domain.

A \define{labelled system} is a labelled space $\mathcal  W$ equipped with a serial, fully confluent, convex sensible relation $R_\mathcal W\subseteq |\mathcal  W|\times|\mathcal  W|$.
If moreover $R_\mathcal W$ is $\omega$-sensible, we say that $\mathcal  W$ is a \define{$\Sigma$-quasimodel.}

\end{defn}

Any bi-relational model can be regarded as a $\Sigma$-qua\-simodel: If $\mathcal {X} = (W,T,{\leq},S, \val \cdot)$ is a bi-relational model and $x\in W \times T$, we can assign a $\Sigma$-type $\ell_\mathcal  X(x)$ to $x$ given by
\begin{align*}
\ell_\mathcal  X(x)^+ &=\cbra\psi\in \Sigma \mid x\in \val\psi \cket\\
\ell_\mathcal  X(x)^- &=\cbra\psi\in \Sigma \mid x\not\in \val\psi \cket.
\end{align*}
Note that this assignment of types is $\leq_\Sigma$-monotone.
We also set $R_\mathcal  X=\{((w, t) , (w, S(t))) \mid w \in W, t \in T\}$; it is obvious that $R_\mathcal  X$ is $\omega$-sensible.
Henceforth we will tacitly identify $\mathcal  X$ with its associated $\Sigma$-quasimodel.

The following is proved in \cite{gtllics}.

\begin{theorem}\label{theoLICS}
Given $\varphi\in\lanfull$, the following are equivalent:
\begin{enumerate}

\item $\varphi$ is falsifiable.

\item $\varphi$ is falsifiable in a quasimodel.

\item $\varphi$ is falsifiable in a finite quasimodel.

\end{enumerate}
\end{theorem}

\skipprf{
For our purposes, we only need to know that a formula falsifiable on a quasimodel is falsifiable on a bi-relational model.
The main issue is that a quasimodel $(W,\leq,R,\ell)$ may be nondeterministic, in the sense that worlds could have multiple $R$-successors.
However, we obtain a deterministic structure by considering paths $w_0 \mathrel R w_1\mathrel Rw_2\mathrel R\ldots $.
The sensibility conditions ensure that such paths may be chosen to obey the semantics of the temporal operators; these paths can then be arranged together to form a bi-relational model.
For precise details of the construction, we refer the reader to \cite{gtllics}.
}

\section{The canonical model}\label{secCan}

In this section we construct a standard canonical model for $\sf GTL$.
In the presence of $\ps$ and $\nec$, the standard canonical model is only a labelled system, rather than a proper bi-relational model.
Nevertheless, it will be a useful ingredient in our completeness proof.
Since we are working over an intermediate logic, the role of maximal consistent sets will be played by complete types, as defined below.
The notation $\vdash$ always refers to derivability in the calculus defined in Section \ref{secAx}.
Below, recall that by convention, $\bigwedge\varnothing = \top$ and $\bigvee\varnothing=\bot$.

\begin{defn}\label{def:complete}
Given two sets of formulas $\Gamma,\Delta \subseteq\lanfull$, we say that $\Delta$ is a consequence of $\Gamma$, denoted by $\Gamma \vdash \Delta$, if there exist finite (possibly empty) $\Gamma'\subseteq \Gamma$ and $\Delta' \subseteq \Delta$ such that $ \vdash \bigwedge \Gamma' \imp \bigvee \Delta'$ (i.e. $\bigwedge \Gamma' \imp \bigvee \Delta' \in \gtl$).

We say that a pair of sets $\Phi =(\Phi^+,\Phi^-) \in \lanfull \times \lanfull$ is \define{consistent} if $\Phi^+ \not\vdash \Phi^-$. A saturated, consistent pair is a \define{complete type}. The set of complete types will be denoted $\ptype{}$.
\end{defn}

Note that we are using the standard interpretation of $\Gamma \vdash \Delta$ in Gentzen-style calculi. When working within a turnstile, we will follow the usual proof-theoretic conventions of writing $\Gamma,\Delta$ instead of $\Gamma \cup \Delta$, and writing $\varphi$ instead of $\{\varphi\}$.
Observe that there is no clash in terminology regarding the use of the word {\em type}.

\begin{lem}\label{lemmcompleteIsType}
If $\Phi$ is a complete  type then $\Phi$ is a saturated two-sided $\lanfull$-type.
\end{lem}

\begin{proof}
Let $\Phi$ be a complete type. Observe that $\Phi$ is already saturated by definition, so it remains to check that it satisfies all conditions of Definition \ref{def:type}.
Condition \ref{cond:type:intersection} follows from the consistency of $\Phi$.\skipprf{
The proofs of the other conditions are all similar to each other. We check the ones related to implication and co-implication:
\begin{itemize}
	\item Condition~\ref{cond:type:implication}: assume by contradiction that $\varphi \imp \psi \in \Phi^+$ but $\varphi \not \in \Phi^-$ and $\psi \not \in \Phi^+$. Since $\Phi$ is saturated, $\varphi \in \Phi^+$ and $\psi \in \Phi^-$. By Rule~\ref{ax13MP}, $\psi \in \Phi^+$: a contradiction.
	\item Condition~\ref{cond:type:implication:neg}: assume by contradiction that $\varphi \imp \psi \in \Phi^-$ but $\psi \not \in \Phi^-$. Since $\Phi$ is saturated $\psi \in \Phi^+$. Since $\psi \imp \left(\varphi  \imp \psi\right)$ is intuitionistically valid, we get $\varphi \imp \psi \in \Phi^+$ by Rule~\ref{ax13MP}: a contradiction.
	\item Condition~\ref{cond:type:coimplication:neg}: under the assumption that $\Phi$ is saturated, assume by contradiction that $\varphi \dimp \psi \in \Phi^-$ but $\varphi \in \Phi^+$ and $\psi \in \Phi^-$.  Therefore, $\psi \vee \left(\varphi \dimp \psi\right) \in \Phi^-$. By Rule~\ref{ax13MP} and Axiom~\ref{axco01} we get $\psi\vee \left(\varphi \dimp \psi\right) \in \Phi^+$: a contradiction. 
	\item Condition~\ref{cond:type:coimplication}: Using Lemma \ref{lemmReverse}(\ref{itReverseDimp}, we see that $\vdash \left(\varphi \dimp \psi \right) \imp \varphi$. Therefore, if $\varphi \dimp \psi \in \Phi^+$ we can conclude $\psi \in \Phi^+$ thanks to modus ponens.
\end{itemize}
}
For condition \ref{cond:type:diam} we use Axiom \ref{ax10DiamFix}: if $\diam \varphi \in \Phi^-$ and $\varphi \in \Phi^+$ we would have that $\Phi$ is inconsistent; hence $\varphi \in \Phi^-$. Condition~\ref{cond:type:box} is proved using Axiom~\ref{ax09BoxFix}. The remaining conditions are left to the reader.
\end{proof}
As with maximal consistent sets, complete types satisfy a Lindenbaum property.
Below, if $(\Gamma,\Delta)$ and $(\Gamma',\Delta')$ are pairs of sets of formulas, we say that $(\Gamma',\Delta')$ \textbf{extends} $(\Gamma,\Delta)$ if $\Gamma\subseteq\Gamma'$ and $\Delta\subseteq\Delta'$.

\begin{lem}[Lindenbaum lemma]\label{LemmLind}
Let $\Gamma,\Delta \subseteq \lanfull$. If $\Gamma \not\vdash \Delta$, then there exists a complete type $\Phi$ extending $(\Gamma,\Delta)$.
\end{lem}

\proof
The proof is standard, but we provide a sketch.
Let $\varphi \in \cl L $. Note that either $\Gamma ,\varphi  \not\vdash \Delta $ or $\Gamma  \not \vdash \Delta,\varphi$, for otherwise by a cut rule (which is intuitionistically derivable) we would have $\Gamma \vdash \Delta$. Thus we can add $\varphi$ to $\Gamma$ or to $\Delta$, and by repeating this process for each element of $\cl L $ (or using Zorn's lemma) we can find a suitable $\Phi$.
\endproof

Before defining the canonical model, recall that for a set of formulas $\Gamma$, we have $\tnext \Gamma \eqdef \lbrace \tnext \varphi \mid   \varphi \in \Gamma\rbrace$.
We also define
\[\circop\Gamma \eqdef \lbrace  \varphi \mid \tnext \varphi \in \Gamma\rbrace.\]

Given a set $A$, let $\mathbb I_A$ denote the identity function on $A$.
The canonical model $\CMod$ is defined as the labelled structure
\[
\CMod = (|\CMod|,{\peq_\CIcon },S_\CIcon ,\ell_\CIcon ),
\]
where $|\CMod| = \ptype{}$ is the set of complete types, $\Phi \peq_\CIcon \Psi$ if $\Phi \peq_\lanfull \Psi$ (i.e., if  $\Phi^- \subseteq \Psi^-$ and $\Phi^+ \supseteq \Psi^+$), $S_\CIcon(\Phi) = (\circop \Phi^+, \circop \Phi^-)$, and $\ell_\CIcon(\Phi)=\Phi$.
We will usually omit writing $\ell_\CIcon $, as it has no effect on its argument.

Next we show that $\CMod$ is an $\lanfull$-labelled system.
We begin by showing that it is based on a labelled space.
\begin{lem}
	\label{lemm:normality} $(|\CMod|,\peq_\CIcon,\ell_\CIcon)$ is a $\lanfull$-labelled space.
\end{lem}

\begin{proof}
	We know that $\leq_\lanfull$ is a partial order and restrictions of partial orders are partial orders, so $\peq_\CIcon $ is a partial order. Moreover, $\ell_\CIcon $ is the identity, so $\Phi \peq_\CIcon  \Psi$ implies that $\ell_\CIcon  (\Phi) \leq_\lanfull \ell_\CIcon  (\Psi)$.
	
	To prove that $(|\CMod|,{\peq_\CIcon })$ is locally linear, assume towards a contradiction that it is not. We consider two cases:
	\begin{enumerate}[wide, labelwidth=!, labelindent=0pt]
		\item There exist $\Phi$, $\Psi$ and $\Theta$
		such that $\Phi \leq_\lanfull \Psi$ and $\Phi \leq_\lanfull \Theta$, but $\Psi \not  \leq_\lanfull \Theta$ and $\Theta \not  \leq_\lanfull \Psi$. By definition, there exist two formulas $\varphi \in \Theta^+ \setminus \Psi^+$ and $\psi \in \Psi^+ \setminus \Theta^+$. 
		It is easy to see that $\varphi \imp \psi \not \in \Theta^+$ and $\psi \imp \varphi \not \in \Psi^+$. This would imply that Axiom~\ref{axGodel} does not belong to $\Phi^+$---a contradiction.
		\item There exist $\Phi$, $\Psi$ and $\Theta$
		such that $\Phi \geq_\lanfull \Psi$ and $\Phi \geq_\lanfull \Theta$, but $\Psi \not  \geq_\lanfull \Theta$ and $\Theta \not  \geq_\lanfull \Psi$. Then it is easy to see that there exist two formulas $\varphi \in \Psi^+ \setminus \Theta^+$ and $\psi \in \Theta^+ \setminus \Psi^+$ such that $\varphi \dimp \psi\in \Psi^+$ and $\psi \dimp \varphi \in \Theta^+$. From $\Phi \geq_\lanfull \Psi$, $\Phi \geq_\lanfull \Theta$, and some intuitionistic reasoning we conclude that $\left( \varphi \dimp \psi \right)\wedge  \left(\psi \dimp \varphi\right) \in \Phi^+$, which contradicts Axiom~\ref{axcoGodel}.
	\end{enumerate}
	We finish by considering the conditions on $\imp$ and $\dimp$. Let us consider $\Phi \in |\CMod|$: 
	\begin{itemize}[wide, labelwidth=!, labelindent=0pt]
		\item If $\varphi \imp \psi \in \Phi^-$ then, by Condition~\ref{cond:type:implication:neg} of Definition~\ref{def:type}, $\psi \in \Phi^-$. Let us define $u = (\Phi^+ \cup \lbrace \varphi \rbrace, \lbrace \psi \rbrace)$, and let us assume by contradiction that $u$ is not consistent. This means that there exists $\gamma\in \Phi^+$ such that $\gamma \wedge \varphi \imp \psi \in \gtl$.
		By propositional reasoning, $\gamma\imp \left(\varphi \imp \psi\right) \in \gtl$. Since $\gamma \in \Phi^+$ and $\Phi$ is consistent, $\varphi \imp \psi \not\in \Phi^-$---a contradiction. Therefore, $u$ is consistent and, by Lemma~\ref{LemmLind}, it can be extended to a complete type $\Psi$. From the definition of $u$ we can conclude that $\Psi \leq_\lanfull \Phi$, $\varphi \in \Psi^+$, and $\psi \in \Psi^-$ as required.

		\item If $\varphi \dimp \psi \in \Phi^+$ then, by Condition~\ref{cond:type:coimplication}, $\varphi \in \Phi^+$. Let us define $u = (\lbrace \varphi\rbrace, \Phi^-\cup \lbrace \psi\rbrace)$, and let us assume by contradiction that $u$ is not consistent. This means that there exists $\gamma \in \Phi^-$ such that $\varphi \imp \psi \vee \gamma \in \gtl$. By Rule \ref{axDimpDis}, we get  $\left(\varphi \dimp \psi\right)\imp \gamma \in \gtl$. Since $\gamma \in \Phi^-$, we deduce that $\varphi \dimp \psi \not\in \Phi^+$---a contradiction. By Lemma~\ref{LemmLind}, $u$ can be extended to a complete type $\Psi$. It is easy to check that $\Phi \leq_\lanfull \Psi$, $\varphi \in \Psi^+$, and $\psi \in \Psi^-$ as required.\qedhere
	\end{itemize}

\end{proof}

\begin{lem}\label{lem:welldefined}
	$S_\CIcon \colon |\CMod|\to  |\CMod| $ is well defined.
\end{lem}

\begin{proof}
	Let $\Phi\in |\CMod| $ and $\Psi = S_\CIcon(\Phi)$; we must check that $\Psi \in |\CMod|= \ptype{}$.
	Recall that $\Psi^+=\circop\Phi^+ $ and $\Psi^-=\circop\Phi^- $. 
	To see that $\Psi$ is saturated, let $\varphi \in \lanfull$ be so that $\varphi \not \in \Psi^-$. It follows that $\tnext\varphi \not \in \Phi^-$, but $\Phi$ is saturated, so $\tnext\varphi \in \Phi^+$ and thus $\varphi \in \Psi^+$. Since $\varphi$ was arbitrary, $\Psi^-\cup \Psi^+ = \lanfull$.
	Next we check that $\Psi$ is consistent. If not, let $\Gamma \subseteq \Psi^+$ and $\Delta\subseteq \Psi^-$ be finite and such that $\bigwedge \Gamma \imp \bigvee \Delta \in \gtl$. Using \ref{ax14NecCirc} and \ref{ax05KNext} we see that $\tnext \bigwedge \Gamma \imp  \tnext \bigvee \Delta \in \gtl$, which in view of Lemma \ref{lemmReverse} implies that $ \bigwedge \tnext \Gamma \imp  \bigvee  \tnext \Delta \in \gtl$ as well. But $\tnext \Gamma \subseteq \Phi^+$ and $\tnext \Delta \subseteq \Phi^-$, contradicting the fact that $\Phi$ is consistent. We conclude that $\Psi \in |\CMod|$.
\end{proof}

\begin{lem}\label{lem:confluent}
	$S_\CIcon$ is fully confluent.
\end{lem}
\begin{proof} We check the four conditions:
	\begin{description}[wide, labelwidth=!, labelindent=0pt ]

		\item[\textbf{Forth--down, forth--up:}] Let $\Phi$, $\Psi$ be such that $\Phi \leq_\lanfull \Psi $.
Since $S_\CIcon$ is a function, these properties amount to showing that $S_\CIcon(\Phi) \leq_\lanfull S_\CIcon(\Psi)$.
If $\varphi\in \big (S_\CIcon(\Psi) \big )^+ $ then $\tnext\varphi\in \Psi^+$, which since $\Phi \leq_\lanfull \Psi $ implies that $\tnext \varphi\in \Phi^+$ and hence $\varphi\in \big (S_\CIcon(\Phi) \big )^+$.
Similarly we can check that if $\varphi\in \big (S_\CIcon(\Phi) \big )^- $ then $\varphi\in \big (S_\CIcon(\Psi) \big )^-$, so that $S_\CIcon(\Phi) \leq_\lanfull S_\CIcon(\Psi)$, as needed.
		
		\item[\textbf{Back--up: }] Let $\Phi$, $\Psi$, and $\Theta$  be such that $\Psi = S_\CIcon(\Phi)$ and $  \Psi \leq_\lanfull \Theta$, and let us define $u = (\circop\Theta^+, \Phi^-\cup \circop \Theta^-)$. Assume toward a contradiction that $u$ is not consistent.
		Therefore, there exist $\gamma \in \Phi^-$, $\varphi \in \Theta^+$, and $\psi \in \Theta^-$ such that $\tnext \varphi \imp \left( \gamma \vee \tnext \psi\right) \in \gtl$. By Lemma \ref{lemmReverse}(\ref{itReverseDimp}, we get that $\left( \tnext \varphi \dimp \tnext \psi \right) \imp \gamma \in \gtl$. Since $\gamma \in \Phi^-$ and $\Phi$ is consistent and saturated, we have $\tnext \varphi \dimp \tnext \psi \in \Phi^-$.
		By Axiom~\ref{ax:backward:dual}, $\tnext {\left( \varphi \dimp \psi \right)} \imp \left(\tnext \varphi \dimp \tnext \psi \right) \in \Phi^+$.
		Since $\tnext \varphi \dimp \tnext \psi \in \Phi^-$ and $\Phi$ is consistent and saturated, we have $\tnext \left( \varphi \dimp \psi\right)\in \Phi^-$.
		Since $\Psi=S_\CIcon(\Phi)$, we have $\varphi \dimp \psi \in \Psi^- \subseteq \Theta^-$.
		Thus either $\varphi \in \Theta^-$ or $\psi \in \Theta^+$---a contradiction.
		By Lemma~\ref{LemmLind}, $u$ can be extended to a complete type $\Upsilon$, which satisfies $\Phi \leq_\lanfull \Upsilon $ and $\Theta=S_\CIcon(\Upsilon)$, as required.
		
		\item[\textbf{Back--down: }] Let $\Phi$, $\Psi$, and $\Theta$ be such that $\Psi=S_\CIcon(\Phi)$ and $\Theta   \leq_\lanfull \Psi$ and define $u=(\Phi^+ \cup \circop \Theta^+, \circop \Theta^-)$ and assume that $u$ is not consistent. This means that there exists $\gamma \in \Phi^+$, $\varphi \in \Theta^+$, and $\psi \in \Theta^-$ such that $\gamma \wedge \tnext \varphi \imp \tnext  \psi \in  \gtl$. By propositional reasoning, $\gamma \imp \left(\tnext \varphi \imp \tnext  \psi\right) \in \gtl$, so $\tnext \varphi \imp \tnext \psi \in \Phi^+$. By Axiom~\ref{ax05KNext}, $\tnext \left(\varphi \imp \psi\right) \in \Phi^+$. Since $\Psi=S_\CIcon(\Phi) $, we have $\varphi \imp \psi \in \Psi^+\subseteq \Theta^+$. Therefore $\psi \in \Theta^+$---a contradiction. By Lemma~\ref{LemmLind}, $u$ can be extended to a complete type $\Upsilon$. It can be checked that $\Upsilon \leq_\lanfull \Phi$ and $\Theta=S_\CIcon(\Upsilon)$, as required.\qedhere
	\end{description}
\end{proof}

%

\begin{lem}\label{lem:convex}
$S_\CIcon$ is a convex relation.	
\end{lem}
\begin{proof}
Since $S_\CIcon$ is a function, images of points are singletons, hence automatically convex.
Thus we need only prove that preimages are convex.
We proceed by contradiction. Let us take $\Upsilon \in |\CMod|$ and let us define $\mathcal{A} = S^{-1}_\CIcon(\Upsilon)$ and let us assume that $\mathcal{A}$	 is not convex. This means that there exist $\Phi,\Psi, \Theta \in |\CMod|$ such that $\Phi, \Psi \in \mathcal{A}$ and $\Phi \leq_\lanfull \Theta \leq_\lanfull \Psi$, but $\Theta \not \in \mathcal{A}$. 
Since $\Phi, \Psi \in \mathcal{A}$ and $\Theta \not \in \mathcal{A}$, it follows that $S_\CIcon(\Phi) = S_\CIcon(\Psi) = \Upsilon\neq S_\CIcon(\Theta)$.
We consider two cases: 
\begin{itemize}
	\item there exists $\tnext \varphi \in \Theta^+$ such that $\varphi \not \in \Upsilon^+$. Then $\tnext \varphi \in \Phi^+ $, so $ \varphi \in \Upsilon^+$---a contradiction. 
	\item there exists $\tnext \varphi \in \Theta^-$ such that $\varphi \not \in \Upsilon^-$. Then $\tnext \varphi \in \Psi^- $ so $  \varphi \in \Upsilon^-$---a contradiction. \qedhere
\end{itemize}
\end{proof}

\begin{lem}\label{lem:sensible}
	$S_\CIcon$ is sensible.
\end{lem}
\begin{proof}

Let us consider $\Phi, \Psi$ such that $\Psi=S_\CMod(\Phi)$. We consider the conditions for $(\Phi,\Psi)$ to be sensible.

If $\tnext \varphi \in \Phi^+$ then $\varphi \in \Psi^+$ by the definition of $S_\CIcon$. 
If $\tnext \varphi \not \in \Phi^+$ then $\tnext\varphi \in \Phi^-$ and, by definition, $\varphi \in \Psi^-$.

If $\diam \varphi \in \Phi^+$ and $\varphi \not \in \Phi^+$, it follows that $\varphi \in \Phi^-$. By Lemma \ref{lemmReverse}, $\diam\varphi \imp \varphi \vee \tnext \diam \varphi \in \gtl$, so we cannot have that $\tnext \diam \varphi \in \Phi^-$, and hence $\tnext \diam \varphi \in \Phi^+$, so that $\diam \varphi \in \Psi^+$. Similarly, if $\diam \varphi \in \Phi^-$ we have that $\tnext\diam\varphi \in \Phi^-$, for otherwise we obtain a contradiction from \ref{ax10DiamFix}. Therefore, $\diam\varphi \in \Psi^-$ as well.

If $\nec \varphi \in \Phi^+$ then, by Axiom~\ref{ax09BoxFix} we get $\varphi, \tnext \nec \varphi \in \Phi^+$. Since $\Psi = S_\CIcon(\Phi)$, we get $\nec \varphi \in \Psi^+$.
Conversely, assume that $\nec \varphi \in \Phi^-$. By Lemma~\ref{lemmReverse}, $\varphi \wedge \tnext \nec \varphi \in \Phi^-$, so either $\varphi \in \Phi^-$ or $\tnext \nec \varphi \in \Phi^-$ (giving in the second case $\nec \varphi \in \Psi^-$). In either case we reach the desired conclusion. 
\end{proof}

We remark the general fact that given a $\Sigma_1$-labelled system and a subformula-closed $\Sigma_2 \subseteq \Sigma_1$, one can restrict the labelling to $\Sigma_2$ in the natural way (by replacing its output at any point by its intersection with $\Sigma_2$).
Doing so yields a $\Sigma_2$-labelled system. This is easily verifiable from the definitions.

\begin{prop}\label{prop:CisW}
The canonical model $\CMod$ is an $\lanfull$-labelled system. Restricting the labelling to any subformula-closed $\Sigma \subseteq \lanfull$ yields a $\Sigma$-labelled system.
\end{prop}

\proof
For the first claim, we need for the following three properties to hold:
\begin{enumerate*}
\item $(|\CMod|,{\peq}_\CIcon ,\ell_\CIcon )$ is a labelled space;
\item $S_\CIcon $ is a serial, fully confluent, convex sensible relation; and
\item $\ell_\CIcon $ has $\type{\cl L }$ as its codomain.
\end{enumerate*}
The first item is Lemma \ref{lemm:normality}.  $S_\CIcon $ is serial since it is a well defined function by Lemma \ref{lem:welldefined}, and it is a fully confluent, convex, sensible relation by Lemmas \ref{lem:confluent},~\ref{lem:convex}, and \ref{lem:sensible}.
 Finally, if $\Phi \in |\CMod|$ then $\ell_\CIcon (\Phi) = \Phi$, which is an element of $\type{\cl L }$ by Lemma \ref{lemmcompleteIsType}.
 
The second claim follows from the observation preceding the proposition.
\endproof

%
%

\section{The canonical quasimodel}\label{Sec:quotient}



In this section we describe a finite quotient $\nicefrac{\CMod}\Sigma$ of the canonical labelled system $\CMod$ constructed in \Cref{secCan}, and we show that $\nicefrac{\CMod}\Sigma$ is a $\Sigma$-labelled system. Later, in \Cref{SecComp}, we will show that $\nicefrac{\CMod}\Sigma$ is also $\omega$-sensible and thus a quasimodel.

We obtain $\nicefrac{\CMod}\Sigma$ from $\CMod$ in two steps. First, we will take a bisimulation quotient to obtain a finite $\Sigma$-labelled space equipped with a fully confluent sensible relation. The second step will be to extend the sensible relation to be convex, yielding a finite $\Sigma$-labelled system.

We describe the quotient explicitly, noting afterwards that it is a particular type of bisimulation quotient. The assumption that $\Sigma$ is finite is only needed at the end: if $\Sigma$ is finite then $\nicefrac{\CMod}\Sigma$ will be finite. So for now let $\Sigma$ be any subformula-closed subset of $\lanfull$, and let $\CMod = (|\CMod|,{\peq_\CMod },S_\CMod ,\ell_\CMod )$ be the canonical labelled system, which by Proposition~\ref{prop:CisW} is a $\Sigma$-labelled system when $\ell_\CMod $ is restricted to a $\Sigma$-labelling, which we assume (and henceforth denote by $\ell$).

 For $\Phi\in |\CMod|$, define 
  $L(\Phi) = \cbra \ell(\Psi) \mid \Psi \compa \Phi \cket$. 
   We define the binary relation $\sim$ on $|\CMod|$ by \[\Phi \sim \Psi \iff (\ell(\Phi), L(\Phi)) = (\ell(\Psi), L(\Psi)).\]
If $\Sigma$ is finite, then clearly $|\CMod| / {\sim}$ is finite.

Note that $\sim$ is the largest relation that is simultaneously a bisimulation with respect to the relations $\leq$ and $\geq$, with $\Sigma$ treated as the set of atomic propositions that bisimilar worlds must agree on.

Now define a partial order $\leq_\mathcal Q$ on the equivalence classes $|\CMod| /{\sim}$ of $\sim$ by
\[[\Phi] \leq_\mathcal Q [\Psi] \iff L(\Phi) = L(\Psi)\text{ and }\ell(\Phi) \ge \ell(\Psi),\]
noting that this is well-defined and is indeed a partial order.

Since each set $L(\Phi)$ can be linearly ordered by inclusion and $\ell(\Phi) \in L(\Phi)$, the poset $(|\CMod| / {\sim}, \leq_\mathcal Q)$ is a disjoint union of linear orders. By defining $\ell_\mathcal Q$ by 
\[\ell_\mathcal Q([\Phi]) = \ell(\Phi)\]
we obtain a $\Sigma$-labelled space $(|\CMod| / {\sim}, \leq_\mathcal Q, \ell_\mathcal Q)$; it is not hard to check that this labelling is inversely monotone and that the clauses for $\imp$ and $\dimp$ hold with this labelling.

Now define the binary relation $R_\mathcal Q$ on $|\CMod| / {\sim}$ to be the smallest relation such that $[\Phi] \mathrel R_\mathcal Q [S(\Phi)]$, for all $\Phi \in |\CMod|$. 

\begin{lem}
The relation $R_\mathcal Q$ is fully confluent and sensible.
\end{lem}

\begin{proof}
It is clear that $R_\mathcal Q$ is sensible. For confluence, suppose $[\Phi] \mathrel R_\mathcal Q [S(\Phi)]$. To see that the forth--up condition holds, suppose further that $[\Phi] \leq_\mathcal Q [\Psi]$. Then as $\ell(\Phi) \in L(\Phi) = L(\Psi)$ there is some $\Theta \geq \Phi$ with $[\Psi] = [\Theta]$. Then we have $[\Theta] \mathrel R_\mathcal Q [S(\Theta)]$ and $[S(\Phi)] \leq_\mathcal Q [S(\Theta)]$, as required for the forth--up condition. The proofs of the remaining three confluence conditions are entirely analogous.
\end{proof}

As promised, we now have a $\Sigma$-labelled space equipped with a fully confluent sensible relation. We now transform this labelled space into a $\Sigma$-labelled system by making the additional relation convex by fiat.

Define $R^+_\mathcal Q$ by $X \mathrel R^+_\mathcal Q Y$ if and only if there exist $X_1 \leq_\mathcal Q X \leq_\mathcal Q X_2$ and $Y_1 \leq_\mathcal Q Y \leq_\mathcal Q Y_2$ such that $X_2 \mathrel R_\mathcal Q Y_1$ and $X_1 \mathrel R_\mathcal Q Y_2$. Now define $\nicefrac{\CMod}\Sigma = (|\CMod| / {\sim}, \leq_\mathcal Q, R^+_\mathcal Q, \ell_\mathcal Q)$.

\begin{lem}\label{lemIsLabelled}
The structure $\nicefrac{\CMod}\Sigma$ is a $\Sigma$-labelled system.
\end{lem}

\begin{proof}
We already know that $(|\CMod| / {\sim}, \leq_\mathcal Q, \ell_\mathcal Q)$ is a $\Sigma$-lab\-elled space. First we must check $R^+_\mathcal Q$ is still fully confluent and sensible. 

For the forth--down condition, suppose $X \leq_\mathcal Q X' \mathrel R^+_\mathcal Q Y'$. Then by the definition of $R^+_\mathcal Q$, there are some $X_2 \geq_\mathcal Q X'$ and $Y_1 \leq_\mathcal Q Y'$ such that $X_2 \mathrel R_\mathcal Q Y_1$. 
 Since $X \leq_\mathcal Q X' \leq_\mathcal Q X_2$, by the forth--down condition for $R_\mathcal Q$ there is some $Y \leq_\mathcal Q Y_1$ with $X \mathrel R_\mathcal Q Y$ and therefore $X \mathrel R^+_\mathcal Q Y$. Since $Y \leq_\mathcal Q Y_1 \leq_\mathcal Q Y'$, we are done. 
 The proof that the forth--up condition holds is just the order dual of that for forth--down. The proofs of the back--down and back--up conditions are similar.


To see that $R^+_\mathcal Q$ is sensible, suppose $ X \mathrel R^+_\mathcal Q Y$ and that $\nx \varphi \in \Sigma$. Take $X_1 \leq_\mathcal Q X \leq_\mathcal Q X_2$ and $Y_1 \leq_\mathcal Q Y \leq_\mathcal Q Y_2$ such that $X \mathrel R_\mathcal Q Y_1$. Then
\begin{align*}
\nx \varphi \in \ell_\mathcal Q(X) &\implies \nx \varphi \in \ell_\mathcal Q(X_1)\\
&\implies \phantom{\nx}\varphi \in \ell_\mathcal Q(Y_2)\\
 &\implies \phantom{\nx}\varphi \in \ell_\mathcal Q(Y)\\
 &\implies \phantom{\nx}\varphi \in \ell_\mathcal Q(Y_1)\\
 &\implies \nx \varphi \in \ell_\mathcal Q(X_2)  &\implies \nx \varphi \in \ell_\mathcal Q(X),
\end{align*}
so $ \nx \varphi \in \ell_\mathcal Q(X) \iff  \varphi \in \ell_\mathcal Q(Y)$. The $\ps$ and $\nec$ cases are similar. 

Finally, we show that $R^+_\mathcal Q$ is convex. Firstly, for the image condition, if $X \mathrel R^+_\mathcal Q Y_1$ and $X \mathrel R^+_\mathcal Q Y_2$ with $Y_1 \leq_\mathcal Q Y \leq Y_2$, then by the definition of $R^+_\mathcal Q$ we can find  $X_2 \geq_\mathcal Q X$ and $Y'_1 \leq_\mathcal Q Y_1$ with $X_2 \mathrel R_\mathcal Q Y'_1$, and similarly $X_1 \leq_\mathcal Q X$ and $Y'_2 \geq_\mathcal Q Y_2$ with $X_1 \mathrel R_\mathcal Q Y'_2$. Since then $X_1 \leq_\mathcal Q X \leq_\mathcal Q X_2$ and $Y'_1 \leq_\mathcal Q Y \leq_\mathcal Q Y'_2$, by the definition of $R^+_\mathcal Q$ we conclude that $X \mathrel R^+_\mathcal Q Y$. The preimage condition is completely analogous. This completes the proof that $\nicefrac{\CMod}\Sigma$ is a $\Sigma$-labelled system.
\end{proof}

\begin{lem}\label{quasi:bound}
Suppose $\Sigma$ is finite, and write $\lgt\Sigma$ for its cardinality. Then the height of $\nicefrac{\CMod}\Sigma$ is bounded by $\lgt\Sigma +1$, and the cardinality of the domain $|\CMod| / {\sim}$ of $\nicefrac{\CMod}\Sigma$ is bounded by $(\lgt\Sigma +1)\cdot 2^{\lgt \Sigma (\lgt \Sigma +1)+1}$
\end{lem}

\begin{proof}
Each element of the domain of $\nicefrac{\CMod}\Sigma$ is a pair $(\ell, L)$ where $L$ is a (nonempty) subset of $\raisebox{2pt}{$\wp$} \Sigma$ and $\ell \in L$. Since $L$ is linearly ordered by inclusion, it has height at most $\lgt \Sigma + 1$. There are $(2^{\lgt \Sigma})^i$ subsets of $\raisebox{2pt}{$\wp$} \Sigma$ of size $i$, so there are at most $\sum_{i = 1}^{\lgt \Sigma +1}(2^{\lgt \Sigma})^i$ distinct $L$. The sum is bounded by $2^{\lgt \Sigma (\lgt \Sigma +1)+1}$. The factor of $\lgt\Sigma +1$ corresponds to choice of an $\ell \in L$, for each $L$.
\end{proof}

\color{black}

Thus we have an exponential bound on the size of $\nicefrac{\CMod}\Sigma$.
Later, once we prove $\nicefrac{\CMod}\Sigma$ is a quasimodel, the decidability of $\gtl$ can be inferred from this bound.
See \cite{gtllics} for a more direct proof of decidability using the same quotient construction.
However, for our purposes, it suffices to observe that $\nicefrac{\CMod}\Sigma$ is finite.

\section{Characteristic formulas}\label{SecChar}

Next we show that there exist formulas defining points in the canonical quotient, i.e.~to each $w \in |\nicefrac{\CMod}\Sigma|$ we assign formulas `distinguishing' $w$.
In fact, we need two versions of such formulas, as we can define them to be either true or false outside of the linear component of $w$.
First, we define a formula $\chi^+_\Sigma(w)$ (or $\chi^+(w)$ when $\Sigma$ is clear from context) such that for all $\Gamma \in |\CMod|$, $\chi^+(w) \in \Gamma$ if and only if  $w=[\Gamma']$ for some $\Gamma'\geq\Gamma$.
Dually, we define $\chi^-(w) = \chi^-_\Sigma(w)$ so that for all $\Gamma \in |\CMod|$, $\chi^-(w) \notin \Gamma$ if and only if  $w=[\Gamma']$ for some $\Gamma'\peq\Gamma$.
Compared to \cite{eventually}, these formulas require dual implication, as they must look `up' and `down' the model.
In this section, we write $\cqm\Sigma = (|\cqm\Sigma|,\leq,R,\ell)$.
We will omit subindices on the $\ell$ and $L$ functions.

\begin{defn}
Fix $\Sigma\Subset \lanfull$.
Given $\Delta\in \type\Sigma$, define $\overrightarrow \Delta = \bigwedge\Delta^+\imp \bigvee\Delta^-$ and $\overleftarrow \Delta = \bigwedge\Delta^+\dimp \bigvee\Delta^-$.
Given $w  \in |\nicefrac{\CMod}\Sigma|$, we define a formula $\chi^0(w) = \chi^0_\Sigma(w)$ by
\[\chi^0(w):=     \bigwedge_{\Delta \in L(w)} {\sim} \overrightarrow{\Delta} \wedge  \bigwedge_{\Delta \notin  L(w)} \neg \overleftarrow {\Delta}  . \]
Then define $\chi^+ (w) = \chi^+ _\Sigma(w)$ by
\[\chi^+(w) = \overleftarrow {\ell(w)} \wedge \chi^0(w) \]
and $\chi^- (w) = \chi^- _\Sigma(w)$ by
\[\chi^-(w) = \chi^0(w) \Rightarrow \overrightarrow{\ell(w)}.\]
\end{defn}

\begin{prop}\label{propSimForm}
Given $w \in |\nicefrac{\CMod}\Sigma|$ and $\Gamma\in |\CMod|$,
\begin{enumerate}[label=\arabic*)]

\item\label{simulability:c0} $\chi^0(w)\in \Gamma^+ $ if and only if  $L (\Gamma) = L(w)$,

\item\label{simulability:c1} $\chi^+(w)\in \Gamma ^+$ if and only if  $[\Gamma] \peq w$, and

\item\label{simulability:c2} $\chi^-(w)\in \Gamma^- $ if and only if  $[\Gamma] \geq w$.

\end{enumerate}
\end{prop}

\proof
Let $w \in |\nicefrac{\CMod}\Sigma|$ and $\Gamma\in |\CMod|$.
\medskip

\noindent \ref{simulability:c0}
First assume that  $\chi^0(w)\in \Gamma^+ $, so that $ \bigwedge_{\Delta \in L(w)}{\sim} \overrightarrow \Delta  \in \Gamma^+$ and $   \bigwedge_{\Delta \notin   L(w)}\neg \overleftarrow {\Delta}  \in \Gamma^+$.
Let $\Delta\in L(w)$.
From ${\sim} \overrightarrow \Delta \in \Gamma^+$, we obtain $\Phi \geq \Gamma$ such that $ \overrightarrow {\Delta} \notin \Phi^+$. Hence there is $\Phi_\Delta\leq \Phi $ with $\bigwedge \Delta ^+\in \Phi_\Delta^+$ and $\bigvee \Delta^-\in \Phi_\Delta^-$, i.e.~$\ell (\Phi_\Delta)=\Delta$.
From local linearity we see that $\Phi_\Delta\compa \Gamma$; hence $\Delta=\ell (\Phi_\Delta) \in L (\Gamma)$.

Similarly, if $\Delta \in \type\Sigma\setminus L(w)$, for any $\Psi\leq \Gamma$ we have that $\overleftarrow{\Delta} \notin \Psi^+ $, so that there is no $\Psi_\Delta\geq \Psi$ with $\bigwedge\Delta^+ \in \Psi_\Delta^+$ and $\bigvee \Delta^-\in \Psi_\Delta^-$.
Thus there is no $\Psi_\Delta \compa \Gamma$ with $\bigwedge\Delta^+ \in \Psi_\Delta^+$ and $\bigvee \Delta^-\in \Psi_\Delta^-$, i.e.~$\Delta\notin L (\Gamma)$ (for the $\Psi_\Delta \geq \Gamma$ case, set $\Psi = \Gamma$; for $\Psi_\Delta \leq \Gamma$ set $\Psi = \Psi_\Delta$).

The converse follows by similar reasoning.
Assume that $L (\Gamma) = L(w) $.
Then from $\Delta\in L (\Gamma)$ we readily obtain ${\sim}\overrightarrow{\Delta} \in \Gamma^+$, and similarly from $\Delta \notin L (\Gamma)$ we obtain $\neg \overleftarrow{\Delta} \in \Gamma^+$, from which we obtain by propositional reasoning $\chi^0(w) \in \Gamma^+$.
\medskip

\noindent\ref{simulability:c1}
If $\chi^+(w)\in \Gamma^+$ then $\chi^0(w) \in \Gamma^+$, hence $L (\Gamma) = L(w)$, while $\overleftarrow{\ell(w)} \in \Gamma^+$ implies that there is some $\Gamma'\geq\Gamma$ with $\ell (\Gamma') = \ell(w)$.
This shows that $w = [\Gamma'] \geq [\Gamma]$, as claimed.
\medskip

\noindent\ref{simulability:c2}
This item is similar to the previous, except that we observe that if $L (\Gamma) \neq L(w)$, then $\chi^-(w) \notin \Gamma^-$.
\endproof

\begin{remark}
Note that the formula $\chi^+_\Sigma(w)$ makes essential use of dual implication, as properties of $w\geq [\Gamma]$ do not affect truth values in $\Gamma$ in the language with $\imp$ alone.
In contrast, the formulas $\chi^-_\Sigma$ are similar to the formulas ${\rm Sim}(w)$ of \cite{eventually}, although we remark that dual implication is still needed to describe the full linear component of $w$.
\end{remark}

Next we establish some provable properties of each of $\chi^+_\Sigma$ and $\chi^-_\Sigma$.
We begin with the former.

\begin{prop}\label{propsubplus}
Given $w \in |\nicefrac{\CMod}\Sigma|$ and $\psi\in \Sigma $:
\begin{enumerate}[label=\arabic*)]
	\item\label{itPropsubplOne} If $\psi\in \ell^-({{w}})$, then
$\vdash \chi^+(w)\imp (\chi^+(w) \dimp \psi ).$

	\item\label{itPropsubplOneb} If $\psi\in \ell^+({w})$, then
	$\vdash \chi^+(w) \imp \psi.$

	

	\item\label{itPropsubplFive}
For any $w\in |\nicefrac{\CMod}\Sigma|$,
$\vdash\displaystyle \chi^+(w) \imp \tnext\bigvee _{{{w}} \mathrel R {{{v}}}  } \chi^+(v).$
\end{enumerate}
\end{prop}

\proof
\noindent \ref{itPropsubplOne}
Let $\Gamma\in |\CMod|$ and assume that $\psi\in \ell^-(w)$ and $\chi^+(w) \in \Gamma^+$; by properties of the canonical model, it suffices to show that $ (\chi^+(w) \dimp \psi)\in \Gamma^+ $.
From $\chi^+(w) \in \Gamma$ and Proposition \ref{propSimForm} we obtain $\Delta\leq \Gamma$ such that $[\Delta] = w$, hence $\chi^+(w) \in \Delta^+$ and $\psi\in\Delta^- $, yielding $(\chi^+(w) \dimp \psi) \in \Gamma $.
\smallskip

\noindent \ref{itPropsubplOneb} If $\psi\in \ell^+ ({w})$, as above, let $\Gamma\in |\CMod|$ be such that $\chi^+(w) \in \Gamma^+$, and $\Delta\leq \Gamma$ with $[\Delta] =w$.
Then $\psi\in\Delta^+ $, yielding $\psi\in \Gamma^+$.
\medskip




\smallskip

\noindent \ref{itPropsubplFive} Let $\Gamma$ be such that $\chi^+(w) \in \Gamma$, so that there is $\Delta\leq \Gamma$ with $[\Delta] = w$.
Then $w\mathrel R [S_\CMod (\Delta)]$ by definition, and moreover $ \chi^{+}([S_\CMod (\Delta)]) \in  S_\CMod ( \Delta^+)$ implies that $\nx \chi^{+}([S_\CMod (\Delta)]) \in \Delta^+$.
Thus $\nx \chi^{+}([S_\CMod (\Delta)]) \in \Gamma^+$ by downward persistence, so that $\bigvee_{w\mathrel R  v} \chi^+(v) \in \Gamma^+$.
\endproof

The formula $\chi^-_\Sigma$ behaves `dually', as established below.

\begin{prop}\label{propsubminus}
Given $w \in |\nicefrac{\CMod}\Sigma|$ and $\psi\in \Sigma $:
\begin{enumerate}[label=\arabic*)]
	\item\label{itPropsubOne} If $\psi\in \ell^-({{w}})$, then
$\vdash \psi\imp \chi^-(w)$.

	\item\label{itPropsubOneb} If $\psi\in \ell^+({w})$, then
$\vdash \big (\psi \imp \chi^- (w) \big )\imp \chi^- (w) $.


	\item\label{itPropsubFive}
For any $w\in |\nicefrac{\CMod}\Sigma|$,
$\vdash\displaystyle \tnext\bigwedge _{{{w}} \mathrel R  {{{v}}}  } \chi^-(v) \imp \chi^-(w).$

\end{enumerate}
\end{prop}

\proof
\noindent \ref{itPropsubOne}
Assume that $\psi\in \ell^- ({w}) \cap \Gamma^+$ and write $w=[\Delta]$.
Then $\psi\in \Delta^-$, which means we cannot have $\Gamma \geq \Delta$. Hence Proposition \ref{propSimForm} implies that $\chi^-(w)\notin \Gamma^-$, i.e.~$\chi^-(w)\in \Gamma^+$.
\medskip

\noindent \ref{itPropsubOneb}
Suppose that $\psi\in \ell^+ ({w})$ and proceed to prove the claim by contrapositive.
If $\chi^-(w) \in \Gamma^-$ for some $\Gamma\in |\CMod|$, then there is $\Delta \leq \Gamma$ such that $w=[\Delta]$.
But then $\chi^-(w)\in \Delta^-$ and $\psi\in \Delta^+$, which implies that $\big ( \psi\imp\chi^-(w) \big ) \in \Delta^-$, hence also $\big ( \psi\imp\chi^-(w) \big ) \in \Gamma^-$, as required.
\medskip


\noindent \ref{itPropsubFive}
Proceed by contrapositive.
If $\chi^-(w)\in \Gamma^-$ for some $\Gamma\in |\CMod|$, then there is $\Delta\leq \Gamma$ such that $w=[\Delta]$.
We have that $w \mathrel R [S_\CMod(\Delta)]$ by definition.
Letting $v= [S_\CMod(\Delta)]$, we have that $\chi^-(v) \in  S_\CMod^-(\Delta)$, hence $\nx \chi^-(v) \in   \Delta^-$, and by downward persistence, $\nx \chi^-(v) \in   \Gamma^-$.
Hence $\nx \bigwedge_{w\mathrel R v} \chi^-(v)  \in \Gamma^-$.
\endproof

\section{Completeness}\label{SecComp}

The formulas $\chi^\pm_\Sigma$ are fundamental in our completeness proof; specifically, we will use them to show that $\cqm\Sigma$ is $\omega$-sensible, hence a quasimodel.
Since validity over the class of quasimodels is equivalent to real validity by Theorem \ref{theoLICS}, completeness will follow.
The following lemma is the first step towards establishing $\omega$-sensibility.
Once again, we write $\cqm\Sigma = (|\cqm\Sigma|,\leq,R,\ell)$, and as usual $R^*$ is the transitive, reflexive closure of $R$.

\begin{lem}\label{syntactic}
If $\Sigma\Subset \lanfull$ and ${{w}}\in|\cqm\Sigma|$, then
\begin{enumerate}

\item $\vdash \bigvee _{w \mathrel R^* v}\chi^+({{{v}}}) \imp \tnext \bigvee _{w \mathrel R^* v}\chi^+(v)  $, and

\item $\vdash \tnext \bigwedge _{w \mathrel R^* v}\chi^-(v)\imp \bigwedge _{w \mathrel R^* v}\chi^-({{{v}}})$.

\end{enumerate}
\end{lem}

\proof
The first item follows from Proposition \ref{propsubplus}(\ref{itPropsubplFive}, as for any $v\in R^*(w)$ we have that
\[\vdash \chi^+({{{v}}}) \imp \tnext \bigvee _{v \mathrel R  u}\chi^+(u)  .\]
Since $v \mathrel R  u$ implies that $w \mathrel R^ *  u$ by transitivity,
\[\vdash \chi^+({{{v}}}) \imp \tnext \bigvee _{w \mathrel R^*  u}\chi^+(u)  .\]
Since $v$ was arbitrary, we obtain
\[\vdash \bigvee _{w \mathrel R^* v}\chi^+({{{v}}}) \imp \tnext \bigvee _{w \mathrel R^* u}\chi^+(u),  \]
which by a change of variables yields the original claim.

Item 2 is similar, but uses Proposition \ref{propsubminus}(\ref{itPropsubFive}.
\endproof

In order to complete our proof that $\cqm\Sigma$ is $\omega$-sensible, it suffices to apply induction to the formulas of Lemma \ref{syntactic}.

\begin{prop}\label{tempinc}\
\begin{enumerate}

\item\label{ittempincone} If ${{w}}\in|\cqm\Sigma|$ and $\diam \psi\in \ell^+ ({{w}})$, then there is ${{{v}}}\in{R^*}({{w}})$ such that $\psi\in \ell^+ ({{{v}}})$.

\item\label{ittempinctwo} If ${{w}}\in|\cqm\Sigma|$ and $\Box \psi\in \ell^- ({{w}})$, then there is ${{{v}}}\in{R^*}({{w}})$ such that $\psi\in \ell^- ({{{v}}})$.

\end{enumerate}
\end{prop}

\proof
\noindent{\bf \ref{ittempincone}.} Towards a contradiction, assume that ${{w}}\in |\cqm\Sigma|$ and $\diam \psi\in \ell^+ ({{w}})$ but, for all ${{{v}}}\in{R^*}({{w}})$, $\psi \in \ell^-({{w}})$.
By Lemma \ref{syntactic}, $\vdash \tnext \bigwedge \limits_{w \mathrel R^* v} \chi^-({{{v}}})\imp \bigwedge\limits_{w \mathrel R^* v} \chi^-({{{v}}})$.
By the $\diam$-induction axiom \ref{ax12:ind:2} and standard modal reasoning, $\vdash \diam \bigwedge \limits_{w \mathrel R^* v} \chi^-({{{v}}})\imp \bigwedge\limits_{w \mathrel R^* v} \chi^-({{{v}}})$;
in particular,
\begin{equation}\label{other}
\vdash \diam \bigwedge _{w\mathrel R*v} \chi^-({{{v}}})\imp \chi^-({{w}}).
\end{equation}

Now let ${{{v}}}\in{R^*}({{w}})$.
By Proposition \ref{propsubminus}(\ref{itPropsubOne} and the assumption that $\psi \in \ell^-({{{v}}})$ we have that
$\vdash \psi \imp \chi^-({{{v}}}) $,
and since ${{{v}}}$ was arbitrary,
$\vdash \psi \imp \bigwedge_{w\mathrel R^* v}\chi^-({{{v}}}) $.
Using distributivity \ref{ax07:K:Dual} we further have that
$\vdash \diam \psi \imp \diam \bigwedge_{w\mathrel R^* v}\chi^-({{{v}}})$.
This, along with (\ref{other}), shows that
$\vdash \diam \psi \imp \chi^-({{w}})$. However, by Proposition \ref{propsubplus}(\ref{itPropsubplOneb} and our assumption that $\diam \psi\in \ell^+ ({{w}})$ we have that
$\vdash \big ( \diam \psi \imp \chi^-({{w}}) \big ) \imp \chi^-(w)$.
Hence by modus ponens we obtain $\vdash \chi^-({{w}}).$
Writing $w=[\Gamma]$, Proposition \ref{propSimForm} yields $\chi^-({{w}}) \notin\Gamma^+$, but this contradicts $\vdash \chi^-({{w}})$.
We conclude that there is ${{{v}}}\in{R^*}({{w}})$ with $\psi \in \ell^+({{v}})$, as needed.
\medskip

\noindent{\bf \ref{ittempinctwo}.}
This is similar to the first item, but dualised.
Towards a contradiction, assume that ${{w}}\in |\cqm\Sigma|$ and $\Box \psi\in \ell^- ({{w}})$ but, for all ${{{v}}}\in{R^*}({{w}})$, $\psi \in \ell^+({{w}})$.
By Lemma \ref{syntactic}, $\vdash \bigvee \limits_{w \mathrel R^* v} \chi^+({{{v}}}) \imp \tnext \bigvee \limits_{w \mathrel R^* v} \chi^+({{{v}}})$.
By the $\Box$-induction axiom \ref{ax11:ind:1}, $\vdash \bigvee \limits_{w \mathrel R^* v} \chi^+({{{v}}})\imp\Box \bigvee \limits_{w \mathrel R^* v} \chi^+({{{v}}})$;
in particular,
\begin{equation}\label{otherb}
\vdash\chi^+({{w}}) \imp \Box  \bigvee _{w\mathrel R^*v} \chi^+({{{v}}}) .
\end{equation}

Now let ${{{v}}}\in{R^*}({{w}})$.
By Proposition \ref{propsubplus}(\ref{itPropsubplOneb} and the assumption that $\psi \in \ell^+({{{v}}})$, we have that
$\vdash \chi^+({{{v}}}) \imp \psi $,
and since ${{{v}}}$ was arbitrary,
$\vdash \bigvee_{w\mathrel R^* v}\chi^+({{{v}}})\imp \psi $.
Using distributivity \ref{ax06KBox} we further have that
$\vdash  \Box \bigvee_{w\mathrel R^* v}\chi^+({{{v}}}) \imp \Box \psi$.
This, along with (\ref{otherb}), shows that
\begin{equation}\label{eqPlusBox}
\vdash \chi^+({{w}}) \imp \Box \psi.
\end{equation}
By Proposition \ref{propsubplus}(\ref{itPropsubplOne} and our assumption that $\Box \psi\in \ell^- ({{w}})$ we have that
$\vdash \chi^+(w) \imp \big (  \chi^+({{w}}) \dimp \Box\psi \big )$,
hence by \eqref{eqPlusBox} and Rule \ref{axDimpMon} we obtain $\vdash \chi^+({{w}}) \imp (\Box\psi\dimp\Box\psi).$
In view of Lemma \ref{lemmReverse}.\ref{itNotRef}, this implies that $\chi^+({{w}})$ is contradictory.
Writing $w=[\Gamma]$, Proposition \ref{propSimForm} yields $\chi^+({{w}}) \in\Gamma$, which once again is impossible. We conclude that there is ${{{v}}}\in{R^*}({{w}})$ with $\psi \in \ell^-({{v}})$.
\endproof

\begin{corollary}\label{laststretch}
If $\Sigma\Subset \lanfull$, then $\cqm\Sigma$ is a quasimodel.
\end{corollary}

\proof
By Lemma \ref{lemIsLabelled}, $ \cqm \Sigma$ is based on a labelled system, while by Proposition \ref{tempinc}, $R$ is $\omega$-sensible.
By definition, $ \cqm \Sigma$ is a quasimodel.
\endproof

 We are now ready to prove that our calculus is complete.

\begin{theorem}\label{theocomp}
If $\varphi \in \lanfull$ is valid, then $ \vdash\varphi$.
\end{theorem}

\proof
We prove the contrapositive.
Suppose $\varphi$ is an unprovable formula and let $\Sigma$ be the set of subformulas of $\varphi$.
Since $\varphi$ is unprovable, there is $\Gamma\in |\CMod|$ with $\varphi\in\Gamma^-$. Hence $[\Gamma] \in |\cqm\Sigma|$ is a point in a quasimodel falsifying $\varphi$, so that by Theorem \ref{theoLICS}, $\varphi$ is not valid.
\endproof

 \section{Concluding remarks}\label{SecConc}

We have provided a sound and complete calculus for the G\"odel temporal logic $\sf GTL$.
These results further cement $\sf GTL$ as a privileged logic for fuzzy temporal reasoning and pave the way for a proof-theoretic treatment of these logics.
Among the challenges in this direction is the design of cut-free or cyclic calculi.

In proving our main results, we have developed tools for the treatment of superintuitionistic temporal logics, specifically identifying the usefulness of combining `henceforth' with co-implication.
We believe that this insight will lead to completeness proofs for related logics, including intuitionistic $\sf LTL$, where complete calculi for `eventually' are available, but not so for `henceforth'.
Along these lines, it should be remarked that the techniques of \cite{eventually} should lead to a sound and complete calculus for the logic with $\imp,\nx$ and $\diam$ (but no co-implication or henceforth), although such a result does not follow immediately from the present work.

Another subject that would be worth studying in the  near future is  bisimulation in G\"odel temporal logic. This tool has been used to determine that temporal operators are not interdefinable in the intuitionistic temporal setting~\cite{Balbiani2017,BalbianiTOCL}.
For the class of temporal here-and-there models, `henceforth' is a basic operator that cannot be defined, while `eventually' becomes definable in terms of `henceforth', `next', and implication~\cite{Balbiani2017,BalbianiTOCL}. When introducing co-implication, results on definability exist in the literature: for a combination of the logic of here-and-there, co-implication and the basic modal logic $\sf K$, it has been proven by \cite{BD18} that modal operators become interdefinable.
We do not know if co-implication has the same effect to our G\"odel temporal logic; a negative answer would require a suitable notion of bisimulation preserving both implications, as well as the temporal operators.

In a different direction, logics such as $\sf PDL$ or $\sf CTL$ may also enjoy naturally axiomatizable G\"odel counterparts.
The techniques developed here and by \cite{gtllics} could very well be applicable in these settings.

\paragraph{Acknowledgments.} 
This work has been partially suppor\-ted by FWO-FWF grant G030620N/I4513N (J.P.A. and D.F.D.), FWO grant 3E017319 (J.P.A.), the projects EL4HC and \'etoiles montantes CTASP at R{\'e}gion Pays de la Loire, France (M.D.), the COST action CA-17124 (M.D. and D.F.D.), and SNSF--FWO Lead Agen\-cy Grant 200021L\_196176/G0E2121N (B.M. and D.F.D.).

\bibliographystyle{plain}

\end{document}